\documentclass[letterpaper, 10 pt, conference]{ieeeconf}
\IEEEoverridecommandlockouts	 
\usepackage{cite}
\usepackage{amsmath,amssymb,amsfonts}
\usepackage{algorithmic}
\usepackage[linkcolor=blue, citecolor=blue]{hyperref}
\usepackage{textcomp}
\usepackage{graphicx,color}
\usepackage{mathrsfs}
\usepackage[vlined,ruled]{algorithm2e}
\usepackage{url}
\usepackage{color}
\usepackage{dsfont}
 \usepackage{bbm}
\usepackage{booktabs}
\usepackage{array}
\usepackage[table]{xcolor} 
\usepackage{yfonts}
\usepackage[normalem]{ulem}
\usepackage{arydshln}
\usepackage{multicol}
\usepackage{soul}

\newtheorem{theorem}{Theorem}[section]
\newtheorem{lemma}[theorem]{Lemma}

\newtheorem{remark}{Remark}
\newtheorem{assumption}[theorem]{Assumption}

\newcommand{\setdef}[2]{\{#1 \; : \; #2\}}
\newcommand{\subscr}[2]{{#1}_{\textup{#2}}}

\newcommand{\until}[1]{\{1,\dots,#1\}}

\usepackage{tabularx}
\usepackage{multirow}
\usepackage{makecell}

\newcommand\aamsout{\bgroup\markoverwith{\textcolor{violet}{\rule[0.5ex]{2pt}{1pt}}}\ULon}

\newcommand{\Ker}{\operatorname{Ker}}

\newcommand{\Image}{\operatorname{Im}}

\newcommand{\real}{\mathbb{R}}

\newcommand{\transpose}{\mathsf{T}} %or \top or \intercal
\newcommand{\T}{\mathsf{T}} %or \top or \intercal

\newcommand{\mc}{\mathcal}

\DeclareSymbolFont{bbold}{U}{bbold}{m}{n}
\DeclareSymbolFontAlphabet{\mathbbold}{bbold}

\newcommand\oprocendsymbol{\hbox{$\square$}}
\newcommand\oprocend{\relax\ifmmode\else\unskip\hfill\fi\oprocendsymbol}

\renewcommand{\theenumi}{(\roman{enumi})}

\newcommand*{\QEDA}{\hfill\ensuremath{\blacksquare}}%
\newenvironment{pfof}[1]{\vspace{1ex}\noindent{\itshape Proof of
    #1:}\hspace{0.5em}} {\hfill\QEDA\vspace{1ex}}

\usepackage{mathtools}

\renewcommand{\baselinestretch}{0.96}
\begin{document}
\pagestyle{empty}
\graphicspath{{img/}}

%?Modular learning reduces the  sample complexity of Transferring LQG controls?

\title{\bf \huge Imitation and Transfer Learning for LQG Control}

\author{Taosha Guo, Abed AlRahman Al Makdah, Vishaal Krishnan, and
  Fabio Pasqualetti \thanks{This material is based upon work supported
    in part by awards ONR-N00014-19-1-2264, AFOSR-FA9550-19-1-0235,
    and AFOSR-FA9550-20-1-0140.  T. Guo and F. Pasqualetti are with
    the Department of Mechanical Engineering, and A. A. Al Makdah is
    with the Department of Electrical and Computer Engineering at the
    University of California,
    Riverside. \href{mailto:tguo023@ucr.edu}{\{\texttt{tguo}},\href{mailto:aalmakdah@engr.ucr.edu}{\texttt{aalmakdah}},\href{mailto:fabiopas@engr.ucr.edu}{\texttt{fabiopas\}@engr.ucr.edu}}.
    V. Krishnan is with the School of Engineering and Applied
    Sciences, Harvard University,
    \href{mailto:vkrishnan@seas.harvard.edu}{\texttt{vkrishnan@seas.harvard.edu}}.}}
%    
%    The authors are with the Department of Electrical and Computer
%    Engineering and the Department of Mechanical Engineering at the
%    University of California, Riverside,{\color{red}Add Vishaal's new
%      affiliation and email}
%    \href{mailto:tguo023@ucr.edu}{\{\texttt{tguo}},
%    \href{mailto:aalmakdah@engr.ucr.edu}{\{\texttt{aalmakdah}},
%    \href{mailto:fabiopas@engr.ucr.edu}{\texttt{fabiopas\}@engr.ucr.edu}}.}}
\maketitle
\thispagestyle{empty}
\begin{abstract}
  In this paper we study an imitation and transfer learning setting
  for Linear Quadratic Gaussian (LQG) control, where (i) the system
  dynamics, noise statistics and cost function are unknown and expert
  data is provided (that is, sequences of optimal inputs and outputs)
  to learn the LQG controller, and (ii) multiple control tasks are
  performed for the same system but with different LQG costs. We show
  that the LQG controller can be learned from a set of expert
  trajectories of length $n(l+2)-1$, with $n$ and $l$ the dimension of
  the system state and output, respectively. Further, the controller
  can be decomposed as the product of an estimation matrix, which
  depends only on the system dynamics, and a control matrix, which
  depends on the LQG cost. This data-based separation principle allows
  us to transfer the estimation matrix across different LQG tasks, and
  to reduce the length of the expert trajectories needed to learn the
  LQG controller to~$2n+m-1$ with $m$ the dimension of the inputs (for
  single-input systems with $l=2$, this yields approximately a $50\%$
  reduction of the required expert data).
\end{abstract}

\section{Introduction}\label{sec: introduction}
\smallskip
\noindent

Imitation and transfer learning are popular techniques to learn
optimal policies while reducing the amount of labeled data. In
imitation learning, an agent is given access to samples of expert
(optimal) behavior and seeks to learn a policy that
mimics this behavior. In transfer learning, a model trained on one
task is used as the starting point for a model on a second related
task. The key idea is that certain features learned by the model on
the first task can be used as a general-purpose set of features for
the second task, allowing the model to learn the second task
efficiently. While these techniques have proven useful in multiple
learning scenarios, including image classification and natural
language processing, their use and~utility in control settings have
mostly escaped~scrutiny.

In this paper we investigate the use of imitation and transfer
learning for Linear Quadratic Gaussian (LQG) control,
which seeks a control policy for a stochastic linear
  system that minimizes the expected value of a quadratic function of
  the state and input \cite{KZ-JCD-KG:96}. We consider multiple
control tasks, where the system dynamics are fixed but the quadratic
cost function varies.\footnote{An example of our setting
    is the control
    of autonomous vehicles with cost functions that capture different
    levels of fuel consumption and travel times.} We assume that the
system dynamics, noise statistics, and cost functions are unknown, and
that datasets are available containing optimal input and output
trajectories for the different cost functions (source tasks). The
questions that we answer include whether it is possible to learn the
LQG controllers from expert data, the required size of the dataset, 
and whether the source datasets can be leveraged to learn the
controller for a target task. We show that our data-based controller
enjoys a separation property similar to the well-known separation
principle \cite{KZ-JCD-KG:96}, and that the lower-dimensional
data-based estimation module can be transferred upon changes of the
cost function to reduce the amount of expert data required for control
design.

% with the Linear Quadratic Regulator (LQR) problem having received
% the most attention, \cite{GB-VK-FP:19,FC-GB-FP:21,CDP-PT:19}
%%,MF-RG-SK-MM:18}.

%\smallskip
\noindent
\textbf{Related work.} A number of approaches to direct and indirect
data-driven control have recently been proposed. Most approaches focus
on learning optimal policies from open-loop data for a fixed task and
cost, e.g., see \cite{BR:18,KZ-BH-TB:20,FD-IM:21}.
Differently %LF-YZ-AP-MK:2019
from these works, this paper considers an imitation and transfer
learning framework, where control policies are constructed by
imitating expert demonstrations and transferring information across
multiple, similar control tasks. Multi-task scenarios have received
less attention, with
\cite{LX-LY-GC-SS:22,YC-AMO-FP-EDA:23} and %,HW-LFT-JA:22
\cite{TTZ-KK-BDL:22} being recent exceptions for system identification
and control design, respectively. In \cite{TTZ-KK-BDL:22}, in
particular, the notion of a common lower-dimensional representation
among the tasks is used to reduce the amount of data required for
control design across tasks. Similarly to \cite{TTZ-KK-BDL:22}, this
paper also exploits a lower-dimensional representation for efficient
transfer learning. However, differently from \cite{TTZ-KK-BDL:22} and
leveraging \cite{AAAM-VK-VK-FP:22}, this paper focuses on the LQG
control problem and provides a precise, quantitative characterization
of the lower-dimensional representation for multi-task LQG design from
expert demonstrations, as well as tight bounds on the required
data. This paper also differs from
\cite{SL-KA-BH-AA:20,YZ-LF-MK-NL:21}, which study the sample
complexity of learning LQG controllers in state-space form from
open-loop data.

\noindent
\textbf{Contribution of the paper.} The main contributions of this
paper are as follows. First, we formalize an imitation and transfer
learning setting for LQG control. We show that the LQG controller can
be learned using an optimal input-output trajectory of length
$n(l+2)-1$, where $n$ denotes the dimension of the system and $l$ the
number of outputs. Further, we show that the proposed LQG controller
is unique for the case of single-input systems, while it admits
multiple representations for multi-input systems. Second, we prove the
existence of a data-based separation principle since the data-based
LQG controller can be written as the product of two matrices: the
estimation matrix, which depends only on the system dynamics, and the
controller matrix, which depends on the system dynamics and the
quadratic cost function. Further, for the case of single-input
systems, we show how the estimation matrix can be learned uniquely
using the expert datasets (we also discuss and validate a procedure
for the multi-input case). Third, we show how the data-based
separation principle can be used for transfer learning because the
estimation matrix remains invariant upon changes of the LQG cost
function. By doing so, we show that an expert dataset of length
$2n+m-1$ is sufficient to learn the LQG controller, thus confirming
the benefits of transfer learning also for control design (for
instance, for single input systems with $l=2$, our transfer learning
technique reduces the amount of expert data by about $50\%$). As minor
results, we show that the estimation matrix is of full row rank, thus
suggesting the minimality of the internal representation, and that the
system dimension can be learned using a single expert input-output
trajectory of finite length.

% \smallskip
% \noindent
% \textbf{Organization of the paper.}  The rest of the paper is
% organized as follows. Section \ref{sec: setup} contains our problem
% setting and preliminary results. Section \ref{sec: main results}
% contains our main results, including our data-based separation
% principle and the sample complexity of our imitation and transfer
% learning setting for LQG control. Finally, Section \ref{sec:
%   conclusion} concludes the paper, and the Appendix contains the
% derivations of some our results.

% For notational convenience, we define the following operators:
% \begin{align}
%   \mc H(D_i) &=
%                \begin{bmatrix*}[l]
%                  u_i^* (T_0) & \cdots & u_i^*(\subscr{T}{f} ...)\\
%                  \vdots & \ddots & \vdots\\
%                  u_i^* (T_0+n-1) & \cdots & u_i^*(\subscr{T}{f})\\
%                  y_i^*(T_0) & \cdots & y_i^*(T_0 ...)\\
%                  \vdots & \ddots & \vdots\\
%                  y_i^*(T_0 + n) & \cdots & y_i^*(T_0 ...)
%                \end{bmatrix*}, \text{ and }\\
%   \mc U(D_i) &=
%                \begin{bmatrix*}[l]
%                  u_i^* (T_0 + n) &
%                  \cdots &
%                  u_i^* (\subscr{T}{f}) 
%                \end{bmatrix*}.
% \end{align}

\section{Problem formulation and preliminary results}\label{sec:
  setup}
Consider the discrete-time, linear, time-invariant system
\begin{equation}\label{eq: system}
\begin{aligned}
    x(t+1) &= A x(t) +B u(t)+ w(t),\\
        y(t)  & = Cx(t)+v(t), \qquad t\geq 0,
\end{aligned}
\end{equation}
where $x(t)\in\real^{n}$ denotes the state, $u(t)\in \real^{m}$ the
control input, $y(t)\in\real^{l}$ the measured output, $w(t)$ the
process noise, and $v(t)$ the measurement noise. We assume that the
process and measurement noise sequences are independent at all times
and satisfy $w(t)\sim\mc{N}(0,W)$ and $v(t)\sim\mc{N}(0,V)$,
% $x(0)\sim\mc{N}(0,\Sigma_0)$,
with $W\succeq 0$ and $V\succ 0$.
% and $\Sigma_0 \succeq 0$.
Further, we assume that $(A, B)$ and
$(A, W^{\frac{1}{2}})$ are controllable, and that $(A,C)$ is
observable.
% and $(A,Q^{\frac{1}{2}})$ are observable

For the system \eqref{eq: system}, the Linear Quadratic Gaussian (LQG)
control problem asks for an input that minimizes~the~cost
\begin{align}\label{eq: LQG cost}
  % \mc{J}\triangleq
  \lim_{T\rightarrow \infty}\mathbb{E} \left[
  \frac{1}{T}\Big(\sum_{t=0}^{T-1}x(t)^{\T}Qx(t) +
  u(t)^{\T} R u(t)\Big) \right] , 
\end{align}
where $Q\succeq 0$, $R\succ 0$ are weight matrices and $T$ is the
control horizon. We assume that $(A,Q^{\frac{1}{2}})$ is observable.

As a classic result \cite{KZ-JCD-KG:96}, the optimal input that solves
the LQG problem can be generated by a dynamic controller:
\begin{equation}\label{eq: compensator}
  \begin{aligned}
    \hat{x}(t+1) &= E \hat{x}(t) + F u(t) +G y(t+1),\\
    u(t) & = H \hat{x}(t),
  \end{aligned}
\end{equation}
where the controller matrices $E\in \mathbb{R}^{n\times n}$,
$F \in \mathbb{R}^{n \times m}$, $G \in \mathbb{R}^{n \times l}$ and
$H \in \mathbb{R}^{m \times n}$ can be obtained by combining the
Kalman filter for \eqref{eq: system} with the static controller that
solves the Linear Quadratic Regulator (LQR) problem for \eqref{eq:
  system} with weight matrices $Q$ and $R$ (separation principle). In
this case, $\hat{x}(t) \in \mathbb{R}^n $ denotes the estimate of
$x(t)$ generated by the Kalman filter and the controller matrices
that satisfy
\begin{equation}\label{eq: compensator matrices}
  \begin{alignedat}{3}
    E &= (I-L_{\text{f}} C )A, && \;\;\;\;\;\; &&F = (I- L_{\text{f}} C)B, \\
    G &= L_{\text{f}}, &&  &&H= K_{\text{LQR}} ,
  \end{alignedat}
\end{equation}
% \begin{align}\label{eq: compensator matrices}
%     E &= (I-L_{\text{f}} C )A, && F = (I- L_{\text{f}} C)B, \\
%     G &= L_{\text{f}}, &&H= -K_{\text{LQR}} ,
% \end{align}
where $K_{\text{LQR}}$ and $L_{\text{f}}$ are the LQR and Kalman
gains, respectively. Although different choices are possible, we
assume that the controller \eqref{eq: compensator} uses the matrices
\eqref{eq: compensator matrices} for simplicity and to further
highlight the connections between our results and the classic
separation-based solution to the LQG control problem. Additionally, we
make the following technical assumption.\footnote{This
    assumption is satisfied for generic choices of system
    parameters.}

\begin{assumption}{\bf \emph{(Observability and controllability of the
      controller)}}\label{asmp: compensator observability}
  Let $K_{\text{LQR},i}$ be the $i$-th row of the LQR gain
  $\subscr{K}{LQR}$. Then, the pair $(E, K_{\text{LQR},i})$ is
  observable for every $i \in \until{m}$, and the pair
  $(E,\subscr{L}{f})$ is controllable. \oprocend
\end{assumption}
% \vkmargin{Is there a reason to consider observability for every pair 
% $(E, K_{\text{LQR},i})$ as opposed to simply that of
% $(E, K_{\text{LQR}})$ - which would be more natural? 
% I'm just wondering if it is just a particular way to
% prove Theorem 3.1, or is $(E, K_{\text{LQR},i})$ for every $i$ 
% actually necessary for the results?}
% \margin{I believe it's not necessary but we weren't able to prove the
%   result without it}

\medskip The optimal inputs generated by the dynamic controller
\eqref{eq: compensator} can also be obtained using a static gain and a
finite window of past inputs and outputs \cite{AAAM-VK-VK-FP:22}. In
particular, the optimal LQG inputs $u^*$ satisfy the following
relation:
\begin{align}\label{eq: static LQG}
  u^*(t+n) = \subscr{K}{LQG}
  \begin{bmatrix}
    U_n(t) \\ Y_{n}(t+1)
  \end{bmatrix},
\end{align}
%where $y^*$ is the output of the system \eqref{eq: system} with inputs
%$u^*$
%where $U_{t,n}$ is constructed as \eqref{eq: Hankel vector} using the optimal input sequence $u^*$ generated from the LQG controller \eqref{eq: compensator}, and $Y_{t+1,n}$ the corresponding output sequence $y^*$ from  $\eqref{eq: system}$, and
where
\begin{align}\label{eq: Klqg model}
 \subscr{K}{LQG} & = H  \begin{matrix}
   \begin{bmatrix}
     F_u+E^n F_x^{\dagger}(I-M_u)& F_y-E^n F_x^{\dagger} M_y 
   \end{bmatrix}
   ,
 \end{matrix}
\end{align}
and $U_n(t)$, $Y_n(t+1)$ are constructed as follows from $u^*$ and its
corresponding output $y^*$ from the system \eqref{eq: system},
\begin{align}\label{eq U and Y}
  U_n(t) = 
  \begin{bmatrix}
    u^*(t) \\ \vdots \\ u^*(t + n - 1)
  \end{bmatrix}
  ,
    Y_n(t+1) = 
  \begin{bmatrix}
    y^*(t + 1) \\ \vdots \\ y^*(t+n)
  \end{bmatrix}
  ,
\end{align}
and
% \begin{align}\label{eq: static LQG}
%   u^*(t) = \subscr{K}{LQG}
%   \begin{bmatrix*}[c]
%     u^*(t-n) \\ \vdots \\ u^*(t-1) \\ y^*(t-n+1) \\ \vdots \\ y^*(t)
%   \end{bmatrix*}
%   ,
% \end{align}
\begin{align*}
  M_u & =
        \begin{bmatrix}
          0 & \\
          H F & & \\
          \vdots & \ddots &  & \\
          H E^{n-2} F & \cdots & H F & 0
        \end{bmatrix}
                               ,
                               F_x
                               =
                               \begin{bmatrix}
                                 H\\
                                 H E\\
                                 \vdots\\
                                 H E^{n-1}
                               \end{bmatrix}, 
  \\
  F_u & = 
        \begin{bmatrix}
          E^{n-1}F & \cdots  & F
          \end{bmatrix}, 
\end{align*}
% \vkmargin{Must the entries of $M_u$ contain powers of $H$, like in $(10)$?}
with $M_y$ and $F_y$ constructed in the same way as $M_u$ and
$F_u$ by replacing $F$ by $G$. The expression \eqref{eq: static LQG}
is convenient for learning purposes and it will be at the basis of our
approach. The next technical result will be useful for our derivations
(a proof can be found in the Appendix).
\begin{lemma}{\bf \emph{(Properties of input-output
      matrices)}}\label{lemma: data matrix}
  Let\footnote{This result holds also when the input and output
    sequences are taken from \eqref{eq: compensator} but are not
    generated by the optimal LQG compensator, that is, when the
    compensator is defined with arbitrary matrices $E$, $F$, $G$, and
    $H$.}
  % \begin{align}\label{eq: H}
  %   H_{t,r,c} =
  %   \begin{bmatrix}
  %     u^*(t) & \cdots & u^*(t+c-1)\\
  %     \vdots & \cdots & \vdots\\
  %     u^*(t+r-1) & \cdots & u^*(t+c+r-2)\\
  %     y^*(t) & \cdots & y^*(t+c)\\
  %     \vdots & \cdots & \vdots\\
  %     y^*(t+r) & \cdots & y^*(t+c+r-1)
  %   \end{bmatrix},
                            %     \end{align}
  \begin{align}\label{eq: H}
    H_{r,c} = 
    \begin{bmatrix}
      % U_{t,r,c} \\ Y_{t+1,r,c}
      U_r(t) & \cdots & U_r(t+c-1) \\
      Y_r(t+1) & \cdots & Y_r(t+c) 
    \end{bmatrix},
  \end{align}
  with $r \in \mathbb{N}_{\geq 0}$, $t \in \mathbb{N}_{\geq 0}$, and
  $U_r(t)$, $Y_r(t+1)$ as in \eqref{eq U and Y}. Then
%      % with $u^* \in \real^m$, $y^* \in \real^l$ and
% $t\in \mathbb{N}_{\geq 0}$.  where $u^* \in \real^m$,
% $y^* \in \real^l$, and $t \in \mathbb N_{\ge 0}$.  where $U_{t,n,c}$
% is constructed as \eqref{eq: Hankel} using the control input
% sequence $u \in \real^m $ generated from \eqref{eq: compensator},
% and $Y_{t+1,n,c}$ the corresponding output sequence $y \in \real^l$
% from \eqref{eq: system}.  Then,
  \begin{align*}
    \text{Rank} (H_{r,c}) = \min \{mr+lr,c, n+lr\} .
  \end{align*}
                          %                           and equality is achieved if
                          %                           Furthermore, $\text{Rank} (H_{r,c}) = n + lr$ whenever $c \ge n+lr$,
                          %                           $mr \ge n$, and 
                          %                           $(E, H)$ in \eqref{eq: compensator} is observable.
                          %                           Then, $\text{Rank} (H_{r,c}) \le n + lr$ for any value of $r$
                          %                           and $c$. Further, equality is achieved whenever
                          %                           $c \ge r \ge n$.
\end{lemma}
\medskip

The static LQG controller $\subscr{K}{LQG}$ can be computed using the
static relation \eqref{eq: static LQG} and a sufficiently long, yet
finite, optimal input-output trajectory. In fact, using optimal input
and output sequences, the LQG gain \eqref{eq: static LQG} can be
written as
\begin{align}\label{eq: reconstruct KLQG}
  \subscr{K}{LQG}  = 
  \underbrace{
  \begin{bmatrix}
    u^*(t+n) & \cdots & u^*(t + n + c-1)
  \end{bmatrix}}_{\bar{U}_{c}}
                        H_{n,c}^\dagger,
\end{align}
with $c \ge n (l+1)$. %The data matrix in \eqref{eq: reconstruct
 % KLQG} has a well-characterized rank, as we show in the next
 % result.
Lemma \ref{lemma: data matrix} implies that the data matrix $H_{n,c}$
in \eqref{eq: reconstruct KLQG} is of full row rank for single-input
systems. In this case, the gain $\subscr{K}{LQG}$ is unique and can be
reconstructed exactly from data. On the other hand, $H_{n,c}$ in
\eqref{eq: reconstruct KLQG} loses rank for multi-input systems, in
which case there exists multiple static LQG gains that satisfy the
relation \eqref{eq: static LQG}.
% \textcolor{red}{This property will also affect the results
% in this paper as discussed in Section~\ref{sec: main results}}.

The design of the LQG compensator \eqref{eq: compensator} or
\eqref{eq: static LQG} is typically done using the system model
\eqref{eq: system}, the noise statistics $W$ and $V$, and the cost
matrices $Q$ and $R$. Instead, in this paper we are interested in
solving the LQG problem in a data-driven setting and without
identifying the system matrices. In particular, we consider an
imitation and transfer learning scenario, where an expert provides a
sequence of inputs $u^*$ minimizing the LQG cost \eqref{eq: LQG cost}
and the corresponding outputs $y^*$ of the system \eqref{eq: system},
for different choices of the weight matrices $Q$ and $R$. The question
that we address is to quantify the sample complexity of learning the
LQG controller. We show that, after sufficient training with different
weight matrices, as little as ${2n+m-1}$ expert samples are sufficient
to learn new LQG~controllers.

To formalize the considered problem, let
$D_i = \{u_i^*, y_i^*\}_{t}^{t+T}$ be the sequence of expert (optimal) inputs solving the LQG problem from time $t $
up to time $t+T$ for the matrices $Q_i$ and $R_i$, and the
corresponding outputs of the system \eqref{eq: system}.
% \footnote{To simplify the notation and without affecting generality,
% we assume here that the expert data spans the interval $[0,T]$. The
% generalization to the case where the expert data is collected in an
% arbitrary interval is straightforward.}
We assume that an expert provides optimal sequences for $N$ source LQG
tasks, that is, $D_1, \dots, D_N$, and one optimal sequence
$D_{\text{target}} = \{u^*_{\text{target}}, y^*_{\text{target}}
\}_t^{t+\bar{T}}$ of length $\bar T$, with $(\bar T <T)$, for the
target LQG task with matrices $Q_\text{target}$ and
$R_\text{target}$. The objective is to learn the gain
$\subscr{K}{LQG}$ for the target task using the data $D_1, \dots, D_N$
and $D_{\text{target}}$. In particular, how large should $T$, $\bar T$
and $N$ be for this design to be feasible?  We remark that the system
dynamics and noise statistics are unknown and remain unchanged across
all source and target tasks, and that the weight matrices associated
with the source and target tasks are not known nor provided by the
expert. %Finally, we make the following (mild) assumption.

\section{Separation principle for data-driven LQG control and optimal
  internal representations}\label{sec: main results}
The separation principle states that the LQG controller can be
designed by combining the solution to two separate simpler problems,
namely, an optimal estimation problem and a deterministic optimal
control problem with quadratic cost \cite{KZ-JCD-KG:96}. This insight
is at the basis of the model-based design of the LQG controller. On
the other hand, equation \eqref{eq: static LQG} shows that the LQG
controller can assume a much simpler, static, form, but the expression
lacks interpretability and does not hint to any separation of
estimation and control. The next result bridges this gap by revealing
a decomposition of $\subscr{K}{LQG}$ into independent controller and
estimation~matrices.

\begin{theorem}{\bf \emph{(Data-driven separation
      principle)}}\label{thm: data-driven separation}
  Let $\subscr{K}{LQG}$ be as in \eqref{eq: static LQG}. Then, using
  the notation in Equation~\eqref{eq: compensator},
  
 {\footnotesize
  \begin{align*}
    \begin{aligned}
      \subscr{K}{LQG} = \underbrace{
        \begin{bmatrix}
          \subscr{K}{LQR} & I_m  
        \end{bmatrix}}_{K}
      \underbrace{
        \begin{bmatrix}
          F_u - ( a \otimes I_n) \tilde{M}_u 	\!\!&\!\! F_y - ( a \otimes I_n) \tilde{M}_y\\
           a \otimes I_m \!\!&\!\! 0
        \end{bmatrix} }_{L_{\text{est}}},
    \end{aligned}
  \end{align*}
} where 

{\footnotesize
    \begin{align}\label{eq: Mu My tilde}
      \tilde{M}_u  = 
      \begin{bmatrix}
        0 \\
        F &  0 \\
        \vdots &  \ddots & \ddots  \\
        E^{n-2}F & \cdots &  F &  0
      \end{bmatrix},
                                   \tilde{M}_y  = 
                                   \begin{bmatrix}
                                     0 \\
                                     G &  0 \\
                                     \vdots &  \ddots & \ddots  \\
                                     E^{n-2}G & \cdots &  G &  0
                                   \end{bmatrix},
    \end{align}}
  and $ a =
  \begin{bmatrix}
    a_0 & \dots & a_{n-1}
  \end{bmatrix}$
  such that %$E^{n} = a_{0} I + a_{1}E + \cdots + a_{n-1}E^{n-1}$.
  \begin{align*}
    E^{n} = a_{0} I + a_{1}E + \cdots + a_{n-1}E^{n-1} .
  \end{align*}
\end{theorem}
\smallskip
%\bigskip 
Theorem \ref{thm: data-driven separation} shows how the LQG
gain $\subscr{K}{LQG}$ can be decomposed as the product of
$\subscr{L}{est}$, which depends only on the system matrices, and $K$,
which depends on the LQR gain. Thus, Theorem \ref{thm:
  data-driven separation} shows that the separation principle also
holds when representing the LQG controller as a function of the input
and output sequences. While also interesting as a standalone result,
Theorem \ref{thm: data-driven separation} becomes particularly useful
in our imitation and transfer learning scenario. In fact, since
$\subscr{L}{est}$ is independent of the LQR gain, it is also
independent of the weight matrices $Q$ and $R$. Thus,
$\subscr{L}{est}$ acts as an invariant component of the LQG
controllers across different LQG tasks, which can reduce the amount of
expert data required to construct LQG controllers (since
$\subscr{L}{est}$ is common to all LQG tasks, only the matrix $K$
needs to be reconstructed for the LQG controller of the target
task). % We will exploit this idea to solve the question of transfer
% learning posed in this paper.
We postpone the proof of Theorem
\ref{thm: data-driven separation} to the~Appendix.

The estimation matrix $\subscr{L}{est}$ is connected to the Kalman
filter. To see this, notice that the optimal LQG inputs satisfy
% can be
% obtained from \eqref{eq: compensator} as
\begin{align*}
  u^*(t) = \subscr{K}{LQR} \hat x(t),
\end{align*}
where $\hat x(t)$ denotes the Kalman estimate of the state
$x(t)$. Similarly, using Theorem \ref{thm: data-driven separation} and
\eqref{eq: static LQG} we obtain
\begin{align*}
  \begin{bmatrix}
    \subscr{K}{LQR} & I_m
  \end{bmatrix}
                      \underbrace{\subscr{L}{est}
                      \begin{bmatrix}
                        U_{n}(t-n)\\ Y_{n}(t-n+1)
                      \end{bmatrix}}_{z(t)} = \subscr{K}{LQR} \hat x(t) .
\end{align*}
Thus, as the Kalman estimate provides an $n$-dimensional
\emph{internal representation} $\hat x(t)$ that can be used to
generate optimal LQG inputs via the LQR gain in a model-based setting,
the estimation matrix $\subscr{L}{est}$ provides an
$(n+m)$-dimensional internal representation, $z(t)$, that can be used
to generate optimal LQG inputs in a data-driven setting
(given the matrix $K$). Yet, while the Kalman filter
uses the whole history of inputs and outputs to generate an optimal
internal representation (encoded in the state of the
  dynamic Kalman filter), the estimator $\subscr{L}{est}$ uses only a
finite window of past inputs and outputs. Finally, as we show in the
proof of Theorem \ref{thm: sample complexity}, the matrix
$\subscr{L}{est}$ is of full row rank, thus suggesting that the
internal representation $z(t)$ is of minimal dimension to compute LQG
inputs from~data.

Theorem \ref{thm: data-driven separation} provides a model-based
expression of the control and estimation components that form the LQG
controller. Next,
% since the focus of this paper is in solving imitation and transfer
% learning problems,
we provide an expression of the estimation matrix $\subscr{L}{est}$
given data from a set of source LQG tasks. We start with the case of
single-input systems and discuss the general case afterwards. We make
the following technical assumption. % regarding the source tasks.

%\margin{Does this hold for $m=1$ or in general?}
\begin{assumption}{\bf \emph{(Number and diversity of source
      tasks)}}\label{asmp: task diversity}
  Let $D_1, \dots, D_N$ be the expert trajectories from the source LQG
  tasks with weight matrices
 % $Q_1, R_1,\dots, Q_N,R_N$
 $\{Q_i, R_i\}_{i=1}^N$.
 Let $\subscr{K}{LQG}^i= K^i \subscr{L}{est}$ be the LQG gain
  of the $i$-th source task. Then,
    %\margin{Shouldn't this be $=$? Since the number
   % of columns is $2n+m-1$, I don't think the rank can ever be greater
    %than that}
    \begin{align*}
  \bigcap_{i=1}^N \Ker(K^i) = \{0\}.
  \end{align*}
  \oprocend
\end{assumption}
% Notice that Assumption \ref{asmp: task diversity} imposes a condition on the choice of the 
%bound on the number of source tasks, and thus on the amount of
%available expert data, as well as on the 
%weight matrices of the source tasks. 
We observe that the above condition is typically satisfied in practice
when $N \ge n+m$ and for generic choices of the weight matrices, as it
simply requires the LQR gains to be independent for different choices
of the weight matrices.

%For for notational convenience, we break down the expert trajectory $D_i = \{u_{i}^*, y_{i}^* \}_t^{t+T}$ into the following Hankel matrices for each source task:
% \begin{align}\label{eq: Hankel operator}
  % H_i  = H_{0,n,T-n+1}, \; \; \;\;
%H_{n,c}^i & = 
% \begin{bmatrix}
	 	%  U_{t,r,c} \\ Y_{t+1,r,c}
%	 	  U^i_n(t) & \cdots & U^i_n(t+T-n) \\
%	 	  Y^i_n(t+1) & \cdots & Y^i_n(t+T-n+1) 
%	 \end{bmatrix} \\
	% & = 
    %          \begin{bmatrix*}
    %             u_i^* (t) & \cdots & u_i^*(t+T-n)\\
    %             \vdots & \ddots & \vdots\\
    %            u_i^* (t+n-1) & \cdots & u_i^*(t+T -1)\\
    %             y_i^*(t+1) & \cdots & y_i^*(t+T-n+1)\\
    %             \vdots & \ddots & \vdots\\
    %             y_i^*(t+n) & \cdots & y_i^*(t+T)
    %           \end{bmatrix*}, \\
%\bar{U}_{c}^i & = 
%               \begin{bmatrix*}[l]
%                 u_i^* (t+n) &
%                 \cdots &
%                 u_i^* (t+T) 
%               \end{bmatrix*},
% \end{align}
%where $i$ denotes the $i$-th source task, and $c = T-n+1$.
\begin{theorem}{\bf \emph{(Learning $ \subscr{L}{est}$ when
      $m=1$)}}\label{thm: Lest m1}
  Let $D_1, \dots, D_N$ be the expert trajectories of length
  $T \ge n(l+2)-1$.  Then,
  \begin{align}\label{eq: KL}
    \Ker(\subscr{L}{est}) = \; \bigcap_{i=1}^N \; \Ker(\bar{U}^{i}_{c_s} { H^{i^\dag}_{n,c_s}}),
  \end{align}
  where {$\bar{U}^{i}_{c_s}$} and $H^{i}_{n,c_s}$ are constructed as in
  \eqref{eq: H} using the expert dataset $D_i$ with $c_s =
  T-n+1$.  % $U_i =
  % \begin{bmatrix*}[l]
  %   u_i^* (n) &
  %   \cdots &
  %   u_i^* (T) 
  % \end{bmatrix*}$ and
  % \begin{align}\label{eq: Hi and Ui}
  %   H_i &=
  %   \begin{bmatrix*}[c]
  %                u_i^* (0) & \cdots & u_i^*(T-n-1)\\
  %                \vdots & \ddots & \vdots\\
  %                u_i^* (n-1) & \cdots & u_i^*(T-1)\\
  %                y_i^*(1) & \cdots & y_i^*(T-n)\\
  %                \vdots & \ddots & \vdots\\
  %                y_i^*(n) & \cdots & y_i^*(T)
  %              \end{bmatrix*} .
  % \end{align}
\end{theorem}
\medskip
\begin{proof}
  %\margin{Is the assumption enough to draw this conclusion?}
  Let $\subscr{K}{LQG}^i$ be the LQG controller of the $i$-th
  task. From Theorem \ref{thm: data-driven separation} we have that
  $\subscr{K}{LQG}^i = K^i \subscr{L}{est}$. Then,
  %{\aammargin{is it $K_i$ or $K^i_{\text{LQR}}$?}}
  \begin{align*}
    \Ker (\subscr{K}{LQG}^i) &= \Ker(\subscr{L}{est}) +
                               \subscr{L}{est}^\dag \left( \Image(\subscr{L}{est}) \cap \Ker(K^i) \right)\\
                             &= \Ker(\bar{U}^{i}_{c_s} H^{i^\dag}_{n,c_s}),
  \end{align*}
  where the last equality is due to Lemma \ref{lemma: data matrix}
  (since $m = 1$, $ H^{i}_{n,c_s}$ is of full row rank and
  $\subscr{K}{LQG}^i = \bar{U}^i_{c_s} H^{i^\dag}_{n,c_s}$. The
  claimed statement now follows from Assumption \ref{asmp: task
    diversity}.
  % , as
  % $\bigcap_{i=1}^N \Ker (K^i_{\text{LQR}}) = \{0\}$.
\end{proof}

\medskip From Theorem \ref{thm: Lest m1}, the kernel of the estimation
matrix $\subscr{L}{est}$ can be learned from a finite number of LQG
datasets, with each dataset comprising optimal input and output
trajectories of finite length. Hence, the estimation matrix
$\subscr{L}{est}$ can also be learned up to multiplication by an
invertible matrix using a basis of the orthogonal complement to
$\Ker(\subscr{L}{est})$. That is,
\begin{align*}
  \subscr{L}{est} = P \cdot \underbrace{\text{Basis} \left( \left(
  \bigcap_{i=1}^N \; \Ker(\bar{U}^i_{c_s} H^{i^\dag}_{n,c_s})\right)^\perp \right)^\transpose
  }_{\subscr{\hat L}{est}} ,
\end{align*}
% \vkmargin{Is a $P^{-1}$ missing, i.e.,
% $\subscr{L}{est} = P \subscr{\hat L}{est} P^{-1}$?}  \margin{Should
% be correct:
% $\subscr{K}{LQG} = K \subscr{L}{est} = K P \subscr{\hat L}{est} =
% \tilde K \subscr{\hat L}{est}$} \vkmargin{Can the complement of
% intersection of kernels be written more simply as the union of
% column spaces?}
% \margin{Possibly, but not sure how that will look like. Also, thiX
%   formula reflects the one in Theorem 3.3. We would have to change
%   that as well}
for some invertible matrix $P$. Then, using Theorem \ref{thm:
  data-driven separation} and for any choice of the weight matrices
$Q$ and $R$, the LQG controller for \eqref{eq: system} can always be
written as the product $KP^{-1}\subscr{\hat{L}}{est}$, where only the
matrix $\hat{K}=KP^{-1}$ depends on the weight matrices $Q$ and $R$
and, from Theorem \ref{thm: Lest m1}, the estimation matrix
$\subscr{\hat{L}}{est}$ can be learned given a sufficiently large and
diverse dataset of expert trajectories. These observations imply that
the controller for the target LQG task can be computed by simply
learning the control matrix $\subscr{\hat{K}}{target}$ as a solution
to the linear system %{\aammargin{why do we have a bar on $U$?}}
\begin{align}\label{eq: K L}
  {\bar{U}}^{\text{target}}_{c_t} =
  \subscr{\hat{K}}{target} \subscr{\hat{L}}{est}
  H^{\text{target}}_{n,{c_t}} ,
\end{align}
% \begin{align*}
%   \subscr{K}{target} =
%   \arg\min_K \| K \subscr{L}{est} H_\text{target} - U_\text{target} \| ,
% \end{align*}
where $\bar{U}_{c_t}^{\text{target}}$ and {$H^{\text{target}}_{n,c_t}$
are constructed as in \eqref{eq: H} from the target dataset
$\subscr{D}{target}$, with $c_t = \bar{T}-n+1$.  The next result
%{\aammargin{single trajectory or multiple trajectories?}}
quantifies the length of the expert {trajectory} $D_\text{target}$. We
make the following assumption on the target dataset
\begin{assumption}{\bf \emph{(Persistency of excitation)}}\label{asmp:
    pers excitation} For every value of $c_t$, the target dataset
  satisfies
  \begin{align*}
    \Ker(\subscr{L}{est}) \cap \Image(H^{\text{target}}_{n,c_t}) =
    \{0\} .
  \end{align*} \oprocend
\end{assumption}
We remark that Assumption \ref{asmp: pers excitation} is generically
satisfied since the entries of $H^{\text{target}}_{n,c_t}$ are driven
by the system noise.

% necessary to reconstruct $\subscr{K}{target}$.

\begin{theorem}{\bf \emph{(Length of expert trajectory to learn
      the target LQG controller)}}\label{thm: sample complexity}
  Let $\bar{T}$ be the length of the expert trajectory in
  $D_\text{target}$. The LQG controller for the target task can be
  learned whenever $\bar{T} \ge 2n + m - 1$.
  %Further,
  %\begin{align*}
  %  K_\text{LQG}^\text{target} = \subscr{\bar{U}}{target} (\subscr{L}{est}
  %  {H}_{n,c_t}^{\text{target}})^\dag + \subscr{X}{target}\subscr{Z}{target},
 % \end{align*}
 %  for some matrix $\subscr{X}{target}$, where $\subscr{Z}{target}$ is a basis of the left null space
 % of ${ {H}_{n,c_t}^{\text{target}}}$.
\end{theorem}
\begin{proof}
  We first show that
  $\text{Rank}(F_y - (a \otimes I_n) \tilde{M}_y) = n$, which
  implies that $\subscr{L}{est}$ is of full row rank. With standard
  manipulation, the matrix $F_y - (a \otimes I_n) \tilde{M}_y$
  can be rewritten~as
  
  {\footnotesize
    \begin{align*}
      \underbrace{\begin{bmatrix}
          E^{n-1}  & \cdots & I
        \end{bmatrix}}_{J}
                              \underbrace{
                              \begin{bmatrix}
                                1 \\
                                -a_{n-1} &  1 \\
                                \vdots &  \ddots & \ddots  \\
                                -a_{1}  & \cdots &  -a_{n-1} &  1
                              \end{bmatrix}
                                                                       }_{S}
                                                                       \underbrace{
                                                                       \begin{bmatrix}
                                                                         G \\
                                                                         & \ddots  & \\
                                                                         & &  G
                                                                       \end{bmatrix}}_{O}.
    \end{align*}
  } \!\!Notice that $S$ is invertible and that
  $\text{Rank} (JO) = \text{Rank}(F_y)= n$ due to Assumption
  \ref{asmp: compensator observability}. Then,
  $\text{Rank} (F_y - (a \otimes I_n) \tilde{M}_y) = \text{Rank}(JSO)
  = \text{Rank}(JO)=n$. Due to Assumption \ref{asmp: pers excitation}
  and Lemma \ref{lemma: data matrix}, the matrix
  $H^{\text{target}}_{n,{c_t}}$ has full column rank $n+m$ ($n+m \leq n(l+1)$) when
  $c_t = n+m$ (equivalently, $\bar T = 2n + m -1$) and
  $\Ker(\subscr{
  L}{est}) \cap \Image(H^{\text{target}}_{n,c_t}) =
  \{0\}$. Thus, $\subscr{
  L}{est} H^{\text{target}}_{n,{c_t}}$ is
  invertible, and finally
  $\subscr{K}{target} = \subscr{\bar{U}}{target} (\subscr{
  L}{est}
  {H}_{n,c_t}^{\text{target}})^{-1}$.
\end{proof}
Using \eqref{eq: reconstruct KLQG}, we notice that the LQG controller
can be learned uniquely from a single trajectory of length
$T \ge n(l+2)-1$ for the single-input case, since the data matrix in
\eqref{eq: reconstruct KLQG} becomes of full row rank. This bound
reflects the complexity of learning the LQG controller in an imitation
learning framework. On the other hand, leveraging the separation
principle in Theorem~\ref{thm: data-driven separation}, the matrix
$\subscr{\hat L}{est}$ can be learned from $N \ge n+1$ expert datasets
and used to learn the LQG controller of any target task. By doing so,
Theorem \ref{thm: sample complexity} states that the expert trajectory
of the target task needs only to be of length $2n$ for the
single-input case. This reduced bound reflects the benefits of the
imitation and transfer learning setting, where data from earlier tasks
 is used to solve future LQG tasks. For instance, when $m = 1$ and
$l=2$, the imitation and transfer approach requires about $50\%$ less
expert data compared to the imitation approach~alone.

\begin{remark}{\bf \emph{(Learning the dimension of the system from
      data)}}\label{remark: dimension}
  The reconstruction of the LQG gain in \eqref{eq: reconstruct KLQG}
  and of the estimation matrix in Theorem \ref{thm: Lest m1} requires
  the knowledge of the dimension of the system to properly construct
  the required matrices. If unknown, the dimension of the system can
  be learned by solving the following minimization problem:
  \begin{align*}
    n = \min \setdef{r \in \mathbb{N}}{\text{Rank}(H_{r,r}) =
    \text{Rank}(H_{r+1,r+1})} .
  \end{align*}
  This follows from Lemma \ref{lemma: data matrix}, since the rank of
  $H_{r,r}$ equals $n(l+1)$ when $r = n$ and the value of $l$ can be
  easily inferred from the expert data ($l$ equals the dimension of
  $y^*$).  We note that this remark is also valid for multi-input
  systems. \oprocend
\end{remark}

\begin{remark}{\bf \emph{(Learning $\subscr{L}{est}$ when
      $m>1$)}}\label{remark: multi input}
  {\color{black}
  Lemma \ref{lemma: data matrix} implies that the data matrix in
  \eqref{eq: reconstruct KLQG} is not full row rank when $m>1$ and
  that the LQG gain in \eqref{eq: static LQG} is not unique. Although
  the decomposition in Theorem \ref{thm: data-driven separation} still
  holds, the computation of $\subscr{L}{est}$ from data is more
  involved than the procedure presented in Theorem \ref{thm: Lest
    m1}. Here we discuss two different ways for this computation but,
  in
  the interest of space and clarity, we leave a detailed treatment for
 future research. First, 
 % \textcolor{purple}{\sout{
 %  the data matrices in \eqref{eq: reconstruct
 %    KLQG} can be made full row rank by adding a small amount of noise
 %  to the output of the compensator \eqref{eq: compensator}, for a
 %  finite duration of time. By doing so, $\subscr{L}{est}$ can be
 %  computed as in Theorem \ref{thm: Lest m1} also when $m>1$. Notice
 %  that such noise does not affect the infinite-horizon LQG
 %  cost.Second,}}
  let $\subscr{K}{LQG}^i$ be the LQG controller of the $i$-th source task. Then,
 \begin{align}\label{eq: equality gain kernel}
    \subscr{K}{LQG}^i = \bar{U}^i_{c_s} H^{i^\dag}_{n,c_s} + X_i Z_i = K^i\subscr{L}{est}
  \end{align}
  for some matrix $X_i$, where $Z_i$ is a basis of the left null space
  of ${H^i_{n,c_s}}$. By stacking these expressions together we obtain  \begin{align}\label{eq: kernel matrix}
    \begin{bmatrix}
      K^1 \\ \vdots \\ K^N
    \end{bmatrix}
    \subscr{L}{est} =
    \begin{bmatrix}
      \bar{U}^1_{c_s} H^{i^\dag}_{n,c_s} + X_1 Z_1 \\ \vdots \\ \bar{U}^N_{c_s} H^{N^\dag}_{n,c_s} + X_N Z_N .
    \end{bmatrix}.
  \end{align}
  Since $\subscr{L}{est}$ has $n+m$ rows, where $n$ is obtained from
  Remark \ref{remark: dimension}, the row space of every gain
  $\subscr{K}{LQG}^i$ must belong to the same $(n+m)$-dimensional
  subspace. Then, the matrix on the right hand side of \eqref{eq:
    kernel matrix} must have a left null space of dimension at least
  $mN - (n+m)$ for an appropriate choice of the matrices
  $X_1,\dots,X_N$. This condition can be used to find the matrices
  $X_1,\dots,X_N$ that satisfy \eqref{eq: equality gain kernel} for a
  sufficiently large number $N$. Finally,
  % with Assumption \ref{asmp:
  %     % task diversity} and 
  similar to Theorem \ref{thm: Lest m1},
  % we
  % have
  \begin{align}\label{eq: kernel opti}
    \Ker(\subscr{L}{est}) \!=
    %\! \bigcap_{i=1}^N \Ker(\subscr{K}{LQG}^i)  =\!
    \; \bigcap_{i=1}^N \Ker(U^i_{n,c_s} H^{i^\dag}_{n,c_s} +  X_i Z_i) .
  \end{align}}
  % Finally, using the notation in Theorem \ref{thm: Lest m1} and
  % \eqref{eq: K L}, $\subscr{L}{est}$ can also be computed by solving
  % the following bilinear problem:
  % \begin{align*}
  %   &\min_{K^1,L^1\dots,K^N,L^N} \sum_{i=1}^N \| \bar U^i_{c_s} -  K^i
  %   L^i H^i_{n,c_s}\|\\
  %   &\text{s.t.} \;\; L^i = L^j \text{ for all } i,j\in \until{N} .
  % \end{align*}
 {Second, using the notation in Theorem \ref{thm:
    Lest m1} and \eqref{eq: K L} and the fact that
  $ \text{vec}(U^i_{c_s}) = (H^{i^{\T}}_{n,c_s} \otimes K^i)
  \text{vec}(\subscr{L}{est})$, $\subscr{L}{est}$ can also be computed
  by solving the following bi-linear problem:
  \begin{align}\label{eq: bilinear opti}
    &\min_{L, K^1,\dots,K^N} \sum_{i=1}^N \| \text{vec} (\bar
      U^i_{c_s}) -  (H^{i^{\T}}_{n,c_s} \otimes K^i) \text{vec}(L)\|
      . %\\
%    &\text{s.t.} \;\; L \in \mathbb{R}^{(m+n) \times n(m+l)} .
  \end{align}
  The convergence properties of the two approaches above deserve a
  full discussion that is beyond the scope of this letter; in the
  next section we provide some numerical evidence.
  \oprocend
  }
\end{remark}

\section{Illustrative example}\label{sec: example}
We use the following model of a batch reactor
system that is open-loop unstable: %~\cite{CDP-PT:19}
{\footnotesize
\begin{align}\label{eq: reactor}
  \begin{split}
    A &=
    \begin{bmatrix*}[r]
      1.178 & 0.001 & 0.511  & -0.403\\
      -0.051 & 0.661 & -0.011 & 0.061 \\
      0.076 & 0.335 & 0.560  & 0.382 \\
      0     & 0.335 & 0.089  & 0.849
    \end{bmatrix*}
    ,
   %  \left[ \begin{smallarray}{cccc}
   %    1.178 & 0.001 & 0.511  & -0.403\\
   %    -0.051 & 0.661 & -0.011 & 0.061 \\
   %    0.076 & 0.335 & 0.560  & 0.382 \\
   %    0     & 0.335 & 0.089  & 0.849
   % \end{smallarray} \right],
    B =
    \begin{bmatrix*}[r]
      0.004\\
      0.467\\
      0.213\\
      0.213
    \end{bmatrix*}
    ,
 %    \left[ \begin{smallarray}{c}
 %      0.004\\
 %      0.467\\
 %      0.213\\
 %      0.213
 % \end{smallarray} \right],  
    \\
    C &=
    \begin{bmatrix*}[r]
      -0.44 &  -0.51  &  0.09  & 0.44
    \end{bmatrix*}
    ,
  %    \left[ \begin{smallarray}{cccc}   	
  % -0.44 &  -0.51  &  0.09  & 0.44
  %    \end{smallarray} \right]    ,
  \end{split}
\end{align}}
with process and measurement noise covariance $W = 1.5 I_4$ and
$V = 0.6$. The weight matrices of the target task~are 
 {\footnotesize
  \begin{align*}
    Q_{\text{target}} & =
                        \begin{bmatrix*}[r]
                          6   &  1   &  1   & -3 \\
                          1   &  1   &  0   & -1 \\
                          1   &  0   &  3   & 0 \\
                          -3  &  -1   & 0  &  2 \\
                        \end{bmatrix*}
                        %   \left[ \begin{smallarray}{cccc}
                        %   6   &  1   &  1   & -3 \\
                        %   1   &  1   &  0   & -1 \\
                        %   1   &  0   &  3   & 0 \\
                        %   -3  &  -1   & 0  &  2 \\
                        % \end{smallarray} \right]
    , \text{ and }
    R_{\text{target}} = 1 .
   %   \begin{smallarray}{c}
   %   1
   % \end{smallarray}.
  \end{align*}}
We compare the model-based approach in Theorem \ref{thm: data-driven
  separation} with the data-driven approach in Theorem \ref{thm: Lest
  m1}.  Using \eqref{eq:
Klqg
  model} we obtain %$\subscr{K}{LQG}^{\text{target}}$ as

\vspace{-.2cm}
{\footnotesize
  \begin{align*}
    \subscr{K}{LQG}^{\text{target}} \!=\!
    \begin{bmatrix}
      -0.01  \!\!&\!\!  0.16 \!\!&\!\! -0.54 \!\!&\!\! 1.02 \!\!&\!\! 2.6 \!\!&\!\!  -13.34 \!\!&\!\! 21.25 \!\!&\!\! -10.60
    \end{bmatrix}
                                                                \!.
                                                                % \left[
                                                                %   \begin{smallarray}{cccccccc}
                                                                %     0.08 &  -0.71 &  2.43 & -3.14 & -1.15 &  14.82 & -22.82 & -10.59
                                                                %   \end{smallarray}
                                                                % \right].
  \end{align*}}
For our data-based approach, we have collected expert trajectories
$D_1, \dots, D_N$ of length $ T = 11 $ from $N = 5$ source tasks, with
weighting matrices $Q_i = i I_4$ and $R_i = I_2$ respectively. Using
Theorem \ref{thm: Lest m1}, we compute the estimation matrix
$\subscr{\hat L}{est}$ as

\vspace{-.2cm}
{\footnotesize
\begin{align*}
  \begin{bmatrix*}[r]
   -0.01  \!\!&\!  0.09 \!\!&\!  -0.30  \!\!&\!  0.45  \!\!&\!  0.08  \!\!&\!  -0.42 \!\!&\!  0.65  \!\!&\! -0.30 \\
    0.02  \!\!&\! -0.18 \!\!&\!   0.54 \!\!&\!  -0.61  \!\!&\!  0.07  \!\!&\!  -0.30  \!\!&\! 0.42  \!\!&\! -0.19 \\
    0.05  \!\!&\! -0.34 \!\!&\!   0.57  \!\!&\!  0.56  \!\!&\!  0.12  \!\!&\! -0.28  \!\!&\! -0.10 \!\!&\!  0.38  \\
   -0.04  \!\!&\!  0.23 \!\!&\!  -0.32  \!\!&\! -0.29  \!\!&\!  0.32  \!\!&\! -0.60  \!\!&\! -0.17 \!\!&\!  0.52  \\
   -0.05  \!\!&\!  0.17 \!\!&\!   0.03  \!\!&\! -0.04  \!\!&\! -0.43  \!\!&\!  0.26  \!\!&\!  0.54 \!\!&\!  0.65
 \end{bmatrix*}
%                                                                .
  % \left[
  % \begin{smallarray}{cccccccc}
  %  -0.01  &  0.08 &  -0.27  &  0.39  &  0.03  &  0.55 &  -0.61  & -0.30 \\
  %   0.01  & -0.13 &   0.48 &  -0.72 &  0.01  &  0.27  & -0.36  & -0.17 \\
  %   0.04  & -0.22 &   0.33  &  0.31  & -0.29  & -0.63  & -0.50 &  -0.11\\
  %  -0.06 &   0.37 &  -0.61  & -0.46  & -0.06  & -0.40  & -0.31 &  -0.06 \\
  %  -0.01  &  0.03 & -0.07  & -0.06  & -0.95  &  0.23  &  0.16 &  -0.01
  % \end{smallarray}
  % \right].
\end{align*}}
It can be verified that ${\subscr{K}{LQG}^\text{target}}^\transpose \in
\Image ({\subscr{\hat L}{est}}^\transpose )$, that is, there exists a
matrix $\subscr{\hat{K}}{target}$ such that $\subscr{K}{LQG}^\text{target} =
\subscr{\hat{K}}{target} \subscr{\hat L}{est}$. This verifies that the
estimation matrix $\subscr{\hat L}{est}$ generates an internal
representation from which the LQG inputs can be computed.

Consider now the same system \eqref{eq: reactor} with two inputs, where the new input matrix and its corresponding cost matrix are:
{\footnotesize
\begin{align*}
  B =
  \begin{bmatrix*}[r]
    0.004 & -0.087 \\
    0.467 & 0.001 \\
    0.213 & -0.235 \\
    0.213 & -0.016
  \end{bmatrix*}
 %  \left[
 % \begin{smallarray}{cc} 
 %    0.004 & -0.087 \\
 %    0.467 & 0.001 \\
 %    0.213 & -0.235 \\
 %    0.213 & -0.016
 %  \end{smallarray} \right],
              , \text{ and }
              R_{\text{target}} =
              \begin{bmatrix*}[r]
                1 & 0 \\ 0 & 4 
              \end{bmatrix*}.
              \end{align*}}
                       
                          %            \left[
   %   \begin{smallarray}{cc}
   %   1 & 0 \\ 0 & 4 
   % \end{smallarray} \right].
 %The model based LQG gain for the target task is
%\begin{align*}
%  \subscr{K}{LQG}^\text{model} =
%  \left[
%  \begin{smallarray}{cccccccccccc}
%  	   0.07  & -0.01 &  -0.66 &  0.08 &  2.14 &  -0.12  & -2.72 &  0.02 &   0.33  &  1.74 & -1.27 &  -0.70 \\
 %  0 &  0.07 &  0.01 &  -0.69 & -0.06 &   2.04  &  0.10 &  -2.41 &  0.37  &  -3.20 &   5.33   & 2.43
%  \end{smallarray}
%  \right]
% \end{align*}

We follow the same steps as in the single input case to compute
$\subscr{K}{LQG}^{\text{target}}$ using the model-based approach in
\eqref{eq: Klqg model}, and then following the procedures in Remark
\eqref{remark: multi input}, we compute $\subscr{\hat L}{est}$
{\color{black} using \eqref{eq: kernel opti} and \eqref{eq: bilinear
    opti} respectively. In Fig.~\ref{fig: multi input} we plot the
  error
  $\|\subscr{K}{LQG}^\text{target}(I - \subscr{\hat{L}}{est}^\dagger
  \subscr{\hat{L}}{est}) \|$ for both approaches as the number of
  source tasks increases.} The convergence of the error implies that
$\subscr{\hat L}{est}$ obtained using the methods in Remark
\eqref{remark: multi input} becomes the correct estimation matrix for
the target LQG~controller. \begin{figure}[!t]
  \centering
  \includegraphics[width=0.8\columnwidth,trim={0cm 0cm 0cm
    0cm},clip]{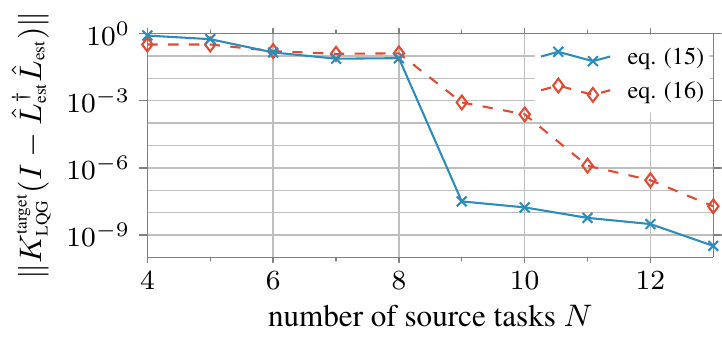}
  \caption{
  {\color{black} This figure shows the error
    $\|\subscr{K}{LQG}^\text{target}(I - \subscr{\hat{L}}{est}^\dagger
    \subscr{\hat{L}}{est}) \|$ as a function of the number of source
    tasks. The error converges for \eqref{eq: kernel opti} and
    \eqref{eq: bilinear opti} as the number of the source tasks
    increases, which implies that both approaches in
    Remark~\eqref{remark: multi input} reconstruct exactly the
    estimation
    matrix.}}
  \label{fig: multi input}
\end{figure}
% \vkmargin{Also, would it help to put the procedure in an algorithm
%   enviroment?}
% \margin{Ideally yes, but the details become messy and we don't have
%   much space. I'd say let's wait for the reviewers to complain... :)}
%
%
\section{Conclusion}\label{sec: conclusion}
In this paper we study an imitation and transfer learning setting for
LQG control, where expert input-output trajectories are used to learn
a data-based LQG controller. We show how the LQG controller can be
computed from data, quantify the length of the expert trajectories
needed to learn the controller, and show how the controller can {\color{black} be}
decomposed as the product of an estimation matrix, which depends only
on the system dynamics, and a controller matrix, which depends also on
the LQG cost. This separation principle allows us to reuse the
estimation matrix across different LQG tasks, thus reducing the length
of the required expert trajectories. Aspects of this research
requiring additional investigation include a detailed treatment of the
multi-input case, the study of transfer methods when the system
dynamics also change, the extension to more general optimal control
problems, and {\color{black}a proof of the minimality of the proposed internal representation}.

\appendix
\subsection{Proof of Lemma \ref{lemma: data matrix}}
\begin{proof}
  From \eqref{eq: compensator} and \eqref{eq U and Y} we have
  {\footnotesize
  \begin{align*}
    U_{r}(t) & =
    \underbrace{
    \begin{bmatrix}
      H \\ H \bar E \\ \vdots \\ H \bar E^{r-1}
    \end{bmatrix}}_{\bar F_x}
    \hat x(t) +
    \underbrace{
    \begin{bmatrix}
      0 & 0 & \cdots  &  0\\
      HG & 0 & \cdots &  0\\
      \vdots & \ddots & \ddots & \vdots\\
      H\bar E^{r-2}G  & \cdots &  HG &0
    \end{bmatrix}}_{\bar F_y}
     	Y_r(t+1),                                
  \end{align*}}
  % \begin{align*}
  %   \underbrace{
  %   \begin{bmatrix}
  %     u^*(t) \\ u^*(t+1) \\ \vdots \\ u^*(t+r-1)
  %   \end{bmatrix}}_{U_t^r}
  %   =
  %   \underbrace{
  %   \begin{bmatrix}
  %     H \\ H \bar E \\ \vdots \\ H \bar E^{r-1}
  %   \end{bmatrix}}_{\bar F_x}
  %   \hat x(t) +& \\
  %   +\underbrace{
  %   \begin{bmatrix}
  %     0 & 0 & \cdots & 0\\
  %     HG & 0 & \cdots & 0\\
  %     \vdots & \ddots & \ddots & \vdots\\
  %     H\bar E^{r-2}G  & \cdots & HG &0
  %   \end{bmatrix}}_{\bar F_y}
  %   &\underbrace{
  %   \begin{bmatrix}
  %     y^*(t) \\ y^*(t + 1) \\ \vdots \\ y^*(t + r - 1)
  %   \end{bmatrix}}_{Y_t^r} ,
  % \end{align*}
  where $\bar E = E + FH$. Then, we obtain
  \begin{align*}
    % \underbrace{
    % \begin{bmatrix}
    %   U_t^r  & \!\!\cdots\!\! & U_{t+c-1}^r\\
    %   Y_t^r & \!\!\cdots\!\! & Y_{t+c}^r
    % \end{bmatrix}}_{H_{t,r,c}}
                               H_{r,c}
                     =
                     \underbrace{
                     \begin{bmatrix}
                       \bar F_x & \bar F_y\\
                       0 & I
                     \end{bmatrix}}_{M}
                           \underbrace{
                           \begin{bmatrix}
                             \hat x(t)  & \!\!\cdots\!\! &
                             \hat x(t+c-1)\\
                             Y_{r}(t+1)  & \!\!\cdots\!\! & Y_{r}(t+c)
                           \end{bmatrix}}_{N} .
  \end{align*}
  % Notice that
  %$\text{Rank}(M) \le \min\{mr+lr,n+lr\}$, and $\text{Rank}(M) = n+lr$ 
  %if $mr \ge n$ and the pair $(\bar E, H)$ is observable. Futher  \cite[Corollary 2]{JCW-PR-IM-BLMDM:05} implies that
  %$\text{Rank}(N) = \min \{n+lr, c\}$. Thus $H_{r,c} \le \min\{\text{Rank}(M),\text{Rank}(N)\}$. When $mr \geq n$ and $c \geq n(l+1)$, since $\text{rank}(N) = n(l+1)$ we have 
  %$\text{rank}(H_{r,c}) = \text{rank}(MN) = \text{rank}(M)= n(n+l)$. This conclude the proof.
  % We conclude that
  % \begin{align*}
  %   \text{Rank}(H_{r,c}) \le \min \{mr+lr, n+lr,c\},
  % \end{align*}
  % and that $\text{Rank}(H_{r,c}) = n+lr$ whenever $c \ge n+lr$ and
  % $mr \ge n$. This concludes the proof.
  Further,
  $\text{Rank} (H_{r,c}) \le \min \{\text{Rank}(M), \text{Rank}(N)\}$,
  and $\text{Rank} (H_{r,c}) = \text{Rank}(M)$  whenever $N$ is of full row rank
  \cite{RAH-CRJ:85}. Notice that
  $\text{Rank}(M) \le \min\{mr+lr,n+lr\}$, and $\text{Rank}(M) = n+lr$
  if $mr \ge n$ and the pair $(\bar E, H)$ is observable. To conclude,
  \cite[Corollary 2]{JCW-PR-IM-BLMDM:05} implies that
  $\text{Rank}(N) = \min \{n+lr, c\}$.
\end{proof}

\subsection{Proof of Theorem \ref{thm: data-driven separation}}
We start with an alternative expression
for~$\subscr{K}{LQG}$.%\margin{The superscript on upper left side is
 % weird. Maybe use $K_{\text{LQR},i}$ and $K_{\text{LQG},i}$}
\begin{lemma}{\bf \emph{(Alternative expression for
      $\subscr{K}{LQG}$)}}\label{lemma: Klqgi}
  Let $K_{\text{LQG},i}$ and $K_{\text{LQR},i}$ be the $i$-th row of
  $\subscr{K}{LQG}$ and $\subscr{K}{LQR}$, respectively, and define
  the matrices $P_i$ such that
  \begin{align*}
    P_ i
    \underbrace{
    \begin{bmatrix}
      \subscr{K}{LQR}\\
      \subscr{K}{LQR} E\\
      \vdots\\
      \subscr{K}{LQR} E^{n-1}
    \end{bmatrix}}_{F_x}
    =
    \underbrace{
    \begin{bmatrix}
      K_{\text{LQR},i} \\ K_{\text{LQR},i} E \\ \vdots \\K_{\text{LQR},i}E^{n-1}
    \end{bmatrix}}_{F_x^i}
    .
  \end{align*}
  for all $i \in \until{m}$. Then,
  {\footnotesize
  \begin{align}\label{eq: klqgi}
  K_{\text{LQG},i} = K_{\text{LQR},i}
    \begin{bmatrix}
      F_u + E^n {F_x^{i}}^\dag (P_i - M_u^i) & F_y - E^n {F_x^i}^\dag
      M_y^i
    \end{bmatrix}                                            ,
  \end{align}}
  where $F_x^i = P_i F_x$, $M_u^i = P_iM_u$, and $M_y^i = P_i M_y$.
\end{lemma}
\begin{proof}
  Using the compensator dynamics \eqref{eq: compensator} we obtain
    \begin{align*}
    %   \underbrace{
    %   \begin{bmatrix}
    %     u(t-n) \\ u(t-n+1) \\ \vdots \\ u(t-1)
    %   \end{bmatrix}}_{U_t}
    U_n(t)
     & =
      F_x
      % &\underbrace{
      %   \begin{bmatrix}
      %     \subscr{K}{LQR}\\
      %     \subscr{K}{LQR} E\\
      %     \vdots\\
      %     \subscr{K}{LQR} E^{n-1}
      %   \end{bmatrix}}_{F_x}
      \hat x(t) +
      M_u
      % \underbrace{
      %   \begin{bmatrix}
      %     0 & 0 & \cdots & 0\\
      %     HF & 0 & \cdots & 0\\
      %     \vdots & \ddots & \ddots & \vdots\\
      %     H E^{n-2}F  & \cdots & HF &0
      %   \end{bmatrix}}_{M_u} 
                                      % \underbrace{
                                      % \begin{bmatrix}
                                      %   u(t-n) \\ u(t-n+1) \\ \vdots \\ u(t-1)
                                      % \end{bmatrix}}_{U_t}
      U_n(t)
      + M_y
      % \\
      % +
      % &\underbrace{
      %   \begin{bmatrix}
      %     0 & 0 & \cdots & 0\\
      %     HG & 0 & \cdots & 0\\
      %     \vdots & \ddots & \ddots & \vdots\\
      %     H E^{n-2}G  & \cdots & HG &0
      %   \end{bmatrix}}_{M_y}
                                      % \underbrace{
                                      % \begin{bmatrix}
                                      %   y(t-n+1) \\ y(t-n+2) \\ \vdots \\ y(t)
                                      % \end{bmatrix}}_{Y_t}
      Y_n(t+1)
      ,
 \end{align*}
 and
 \begin{align}\label{eq: x hat t+n}
    \hat x(t+n) & = E^n \hat x(t)  +
                                  % \underbrace{
                                  % \begin{bmatrix}
                                  %   E^{n-1} F & E^{n-2} F & \cdots & F
                                  % \end{bmatrix}}_{F_u}
                                                                       F_u
                                                                     U_n(t) +
                                                                              %\underbrace{
                                                                              %\begin{bmatrix}
                                                                            %    E^{n-1} G \!&\!\! E^{n-2} G \!&\! \!\!\!\cdots\!\!\! \!&\! G
                                                                            %  \end{bmatrix}}_{F_y}
                                                                            F_y
                                                                          Y_n(t+1) .
  \end{align}
  % By combining the above two equations we obtain
  % \begin{align*}
  %   \hat x(t) = E^n F_x^\dag \left(  (I - M_u)U_t - M_y Y_t  \right)
  %   + F_u U_t + F_y Y_t .
  % \end{align*}  
  Due to Assumption~\ref{asmp: compensator observability},
  ${F} _{x}^i$ is invertible so that
  \begin{align*}
    \hat x (t) = {F_x^i}^\dag \left( 
    (P_i - M_u^i) U_n(t) - M_y^i Y_n(t+1)\right),
  \end{align*}
  for any $i \in \until{m}$ and, consequently,
  {\footnotesize
  \begin{align}\label{eq: xhat}
  \begin{split}
    \hat x (t+n) = E^n {F_x^i}^\dag \big( 
    (P_i - M_u^i) U_n(t) &- M_y^i Y_n(t+1)\big) \\ + F_u U_n(t) + F_y Y_n(t+1) .
    \end{split}
  \end{align}}
  % \bigskip
  % where \todo{Add i to u as well}
  % \begin{align*}
  %   U_t = 
  %   \begin{bmatrix}
  %     u(t-n) \\ u(t-n+1) \\ \vdots \\ u(t-1)
  %   \end{bmatrix}
  %   \text{ and }
  %   Y_t =
  %   \begin{bmatrix}
  %     y(t-n+1) \\ y(t-n+2) \\ \vdots \\ y(t)
  %   \end{bmatrix}
  %   .
  % \end{align*}
Notice that the gain $K_{\text{LQR},i}$ must satisfy, at all times,
  \begin{align*}
    K_{\text{LQR},i} \hat x(t+n) = K_{\text{LQG},i}
    \begin{bmatrix}
      U_n(t) \\ Y_n(t+1)
    \end{bmatrix}
    .
  \end{align*}
  Substituting \eqref{eq: xhat} into $\hat x(t+n)$ in \eqref{eq: x hat
    t+n} yields the result.
\end{proof}

We are now ready to prove Theorem \ref{thm: data-driven separation}.

\begin{pfof}{Theorem \ref{thm: data-driven separation}}
  Notice that \eqref{eq: klqgi} can be rewritten as
  {\footnotesize
  \begin{align*}
    K_{\text{LQG},i} =
    \begin{bmatrix}
      K_{\text{LQR},i} & K_{\text{LQR},i}
    \end{bmatrix}
                          \begin{bmatrix}
                            F_u & F_y\\
                            E^n {F_x^i}^\dag (I-M_u^i) & -E^n {F_x^i}^\dag M_y^i
                          \end{bmatrix}.
  \end{align*}}
Further, using the Cayley-Hamilton Theorem, we have
  \begin{align*}
   K_{\text{LQR},i} E^n &= K_{\text{LQR},i} \left(a_0 I_n +
                            a_{1} E + \cdots + a_{n-1}E^{n-1}_1 \right) \\
                          &=
                          \underbrace{
                            \begin{bmatrix}
                              a_0 & a_1 & \dots & a_{n-1}
                            \end{bmatrix}}_{a}
                                                   F_x^i,
  \end{align*}
  where $a_0,\dots,a_{n-1}$ are the negative coefficients of
  the characteristic polynomial of $E$. Then, since $F_x^i$ is
  invertible (Assumption \ref{asmp: compensator observability}),
  we have
  $K_{\text{LQR},i} E^n {F_x^i}^\dag (P_i-M_u^i) = a
  (P_i-M_u^i)$ and
  $K_{\text{LQR},i} E^n {F_x^i}^\dag M_y^i = a M_y^i$, and
  \eqref{eq: klqgi} becomes
  \begin{align}\label{eq: Klqgi 2}
   K_{\text{LQG},i} =
    \begin{bmatrix}
     K_{\text{LQR},i}& 1
    \end{bmatrix}
                          \begin{bmatrix}
                            F_u & F_y\\
                            a (P_i-M_u^i) & -  a M_y^i
                          \end{bmatrix}.
  \end{align}
  Notice that 
  {\footnotesize
  \begin{align*}
    M_{u}^{i} \!=\!
    \underbrace{
    \begin{bmatrix}
      K_{\text{LQR},i}  \!\!&\!\! \\
      \!\!&\!\!  \ddots  \\ 
      \!\!&\!\! \!\!&\!\! K_{\text{LQR},i}
    \end{bmatrix}}_{\subscr{K}{diag}} \tilde{M}_u,
\;
            M_{y}^{i} \!=\!
            \begin{bmatrix}
              K_{\text{LQR},i}  \!\!&\!\! \\
              \!\!&\!\!  \ddots  \\ 
              \!\!&\!\! \!\!&\!\! K_{\text{LQR},i}
            \end{bmatrix} \tilde{M}_y,
  \end{align*}}
where $\tilde{M}_u$ and $\tilde{M}_y$ are defined in \eqref{eq: Mu My tilde}, and
{\footnotesize
   \begin{align*}
     a
     \subscr{K}{diag} & = 
\begin{bmatrix}
a_0 K_{\text{LQR},i} & \cdots & a_{n-1} K_{\text{LQR},i} \\
\end{bmatrix} 
= K_{\text{LQR},i} \left( a \otimes I_n \right),
   \end{align*}}
 Thus, \eqref{eq: Klqgi 2} becomes 
 {\footnotesize
\begin{align*}
    K_{\text{LQG},i} =
    \begin{bmatrix}
      K_{\text{LQR},i} & 1
    \end{bmatrix}
                          \begin{bmatrix}
                            F_u - (a \otimes I_n) \tilde{M}_u & F_y - (a \otimes I_n) \tilde{M}_y\\
                            a P_i & 0
                          \end{bmatrix}.
  \end{align*}
} By using $\begin{bmatrix}
	 (a P_1)^\transpose & \cdots & (a P_m)^\transpose
       \end{bmatrix}^\transpose = a \otimes I_m,$
% \begin{align*}
% 	\begin{bmatrix}
% 	 cP_1 \\ \vdots \\cP_m
% 	\end{bmatrix} = c \otimes I_m,
% \end{align*}
we obtain
 {\footnotesize
\[
\begin{bmatrix}
	K_{\text{LQG},1} \\ 
	\vdots \\
	K_{\text{LQG},m}
\end{bmatrix}
\!=\!
\begin{bmatrix}
  \subscr{K}{LQR} \!\!&\!\! I_m
\end{bmatrix}
\begin{bmatrix}
  F_u - (a \otimes I_n) \tilde{M}_u \!\!&\!\! F_y - (a \otimes I_n) \tilde{M}_y\\
  a \otimes I_m \!\!&\!\! 0
\end{bmatrix}.
\]} This concludes the proof of Theorem \ref{thm: data-driven
separation}.
\end{pfof}

\bibliographystyle{unsrt}
\bibliography{alias,Main,FP,New}

\end{document}